\documentclass{llncs}

\usepackage[ruled, vlined, nofillcomment, linesnumbered]{algorithm2e}

\usepackage{cite}
\usepackage{url}
\usepackage{amsmath}
\usepackage{amssymb}
\usepackage{graphicx}
\usepackage[tight]{subfigure}
\usepackage[english]{babel}
\usepackage{booktabs}
\usepackage{nicefrac}
\usepackage{verbatim}

\usepackage{tikz}
\usetikzlibrary{arrows,shapes,positioning,fit,calc,matrix,trees}
\usepackage{pgfplots}

\usepackage[pdftitle={Discovering Nested Communities},
           pdfauthor={Nikolaj Tatti and Aristides Gionis},
		   hidelinks,
		   hyperfootnotes = false]{hyperref}

\newcommand{\set}[1]{\left\{#1\right\}}
\newcommand{\pr}[1]{\left(#1\right)}
\newcommand{\fpr}[1]{\mathopen{}\left(#1\right)}

\newcommand{\abs}[1]{{\left|#1\right|}}

\newcommand{\enset}[2]{\left\{#1 ,\ldots , #2\right\}}
\newcommand{\enpr}[2]{\pr{#1 ,\ldots , #2}}
\newcommand{\enlst}[2]{{#1} ,\ldots , {#2}}

\newcommand{\np}{\textbf{NP}}

\newcommand{\define}{\leftarrow}
\newcommand{\reals}{{\mathbb{R}}}

\DeclareRobustCommand{\dispfunc}[2]{%
  \ensuremath{%
  \ifthenelse{\equal{#2}{}}%
    {\mathit{#1}}%
    {\mathit{#1}\fpr{#2}}}}

\newcommand{\wf}{{\ensuremath{w}}}
\newcommand{\wght}[1]{\dispfunc{\wf}{#1}}
\newcommand{\dnst}[1]{\dispfunc{d}{#1}}
\newcommand{\nbhd}[1]{\dispfunc{c}{#1}}

\newcommand{\score}[1]{\dispfunc{q}{#1}}

\newcommand{\fm}[1]{\mathcal{#1}}

\newcommand{\sortnodes}{\textsc{SortVertices}\xspace}

\newcommand{\choosetwo}[1]{{\ensuremath{{#1} \choose 2}}}
\newcommand{\degree}[1]{\dispfunc{\mathrm{deg}}{#1}}
\newcommand{\full}[1]{\ensuremath{{#1}_0}}
\newcommand{\fullwght}[1]{\dispfunc{\wf_0}{#1}}

\newcommand{\wghtmin}[1]{\dispfunc{\wf_m}{#1}}
\newcommand{\wghtsum}[1]{\dispfunc{\wf_s}{#1}}
\newcommand{\wghtnorm}[1]{\dispfunc{\wf_n}{#1}}

\SetKwComment{tcpas}{\{}{\}}
\SetCommentSty{textnormal}
\SetArgSty{textnormal}
\SetKw{False}{false}
\SetKw{True}{true}
\SetKw{Null}{null}
\SetKwInOut{Output}{output}
\SetKwInOut{Input}{input}
\SetKw{AND}{and}
\SetKw{OR}{or}
\SetKw{Break}{break}

\pgfdeclarelayer{background}
\pgfdeclarelayer{foreground}
\pgfsetlayers{background,main,foreground}

\definecolor{yafaxiscolor}{rgb}{0.3, 0.3, 0.3}

\definecolor{yafcolor1}{rgb}{0.4, 0.165, 0.553}
\definecolor{yafcolor2}{rgb}{0.949, 0.482, 0.216}
\definecolor{yafcolor3}{rgb}{0.47, 0.549, 0.306}
\definecolor{yafcolor4}{rgb}{0.925, 0.165, 0.224}
\definecolor{yafcolor5}{rgb}{0.141, 0.345, 0.643}
\definecolor{yafcolor6}{rgb}{0.965, 0.933, 0.267}
\definecolor{yafcolor7}{rgb}{0.627, 0.118, 0.165}
\definecolor{yafcolor8}{rgb}{0.878, 0.475, 0.686}

\newlength{\yafaxispad}
\newlength{\yaftlpad}
\newlength{\yaflabelpad}
\newlength{\yafaxiswidth}
\newlength{\yafticklen}
\makeatletter
\def\pgfplots@drawtickgridlines@INSTALLCLIP@onorientedsurf#1{}
\makeatother

\newcommand{\yafdrawxaxis}[2]{
	\pgfplotstransformcoordinatex{#1}\let\xmincoord=\pgfmathresult 
	\pgfplotstransformcoordinatex{#2}\let\xmaxcoord=\pgfmathresult 
	\pgfsetlinewidth{\yafaxiswidth} 
	\pgfsetcolor{yafaxiscolor}
	\pgfpathmoveto{\pgfpointadd{\pgfpointadd{\pgfplotspointrelaxisxy{0}{0}}{\pgfqpointxy{\xmincoord}{0}}}{\pgfqpoint{-0.5\yafaxiswidth}{\yafaxispad}}}
	\pgfpathlineto{\pgfpointadd{\pgfpointadd{\pgfplotspointrelaxisxy{0}{0}}{\pgfqpointxy{\xmaxcoord}{0}}}{\pgfqpoint{0.5\yafaxiswidth}{\yafaxispad}}}
	\pgfusepath{stroke}

}
\newcommand{\yafdrawyaxis}[2]{
	\pgfplotstransformcoordinatey{#1}\let\ymincoord=\pgfmathresult 
	\pgfplotstransformcoordinatey{#2}\let\ymaxcoord=\pgfmathresult 
	\pgfsetlinewidth{\yafaxiswidth} 
	\pgfsetcolor{yafaxiscolor}
	\pgfpathmoveto{\pgfpointadd{\pgfpointadd{\pgfplotspointrelaxisxy{0}{0}}{\pgfqpointxy{0}{\ymincoord}}}{\pgfqpoint{\yafaxispad}{-0.5\yafaxiswidth}}}
	\pgfpathlineto{\pgfpointadd{\pgfpointadd{\pgfplotspointrelaxisxy{0}{0}}{\pgfqpointxy{0}{\ymaxcoord}}}{\pgfqpoint{\yafaxispad}{0.5\yafaxiswidth}}}
	\pgfusepath{stroke}
}

\newcommand{\yafdrawaxis}[4]{\yafdrawxaxis{#1}{#2}\yafdrawyaxis{#3}{#4}}

\pgfplotscreateplotcyclelist{yaf}{%
{yafcolor1,mark options={scale=0.75},mark=o}, 
{yafcolor2,mark options={scale=0.75},mark=square},
{yafcolor3,mark options={scale=0.75},mark=triangle},
{yafcolor4,mark options={scale=0.75},mark=o},
{yafcolor5,mark options={scale=0.75},mark=o},
{yafcolor6,mark options={scale=0.75},mark=o},
{yafcolor7,mark options={scale=0.75},mark=o},
{yafcolor8,mark options={scale=0.75},mark=o}} 

\pgfplotsset{axis y line=left, axis x line=bottom,
	tick align=outside,
	compat = 1.3,
	tickwidth=\yafticklen,
	clip = false,
	every axis title shift = 0pt,
    x axis line style= {-, line width = 0pt, opacity = 0},
    y axis line style= {-, line width = 0pt, opacity = 0},
    x tick style= {line width = \yafaxiswidth, color=yafaxiscolor, yshift = \yafaxispad},
    y tick style= {line width = \yafaxiswidth, color=yafaxiscolor, xshift = \yafaxispad},
    x tick label style = {font=\scriptsize, yshift = \yaftlpad},
    y tick label style = {font=\scriptsize, xshift = \yaftlpad},
    every axis y label/.style = {at = {(ticklabel cs:0.5)}, rotate=90, anchor=center, font=\scriptsize, yshift = -\yaflabelpad},
    every axis x label/.style = {at = {(ticklabel cs:0.5)}, anchor=center, font=\scriptsize, yshift = \yaflabelpad},
    x tick label style = {font=\scriptsize, yshift = 1pt},
    grid = major,
    major grid style  = {dash pattern = on 1pt off 3 pt},
	every axis plot post/.append style= {line width=\yafaxiswidth} ,
	legend cell align = left,
	legend style = {inner sep = 1pt, cells = {font=\scriptsize}},
	legend image code/.code={%
		\draw[mark repeat=2,mark phase=2,#1] 
		plot coordinates { (0cm,0cm) (0.15cm,0cm) (0.3cm,0cm) };%
	} 
}

\newlength{\adjlength}
\newcommand{\adjbox}[2]{
\pgfmathtruncatemacro{\adjside}{#1 - 1}
\fill[fill = #2, draw = #2, line width = 1pt] (0, #1\adjlength) -- (0, \adjlength)
\foreach \x in {1,..., \adjside} {-- ++(\adjlength, 0) -- ++(0, \adjlength) } -- cycle;
\fill[fill = #2, draw = #2, line width = 1pt] (#1\adjlength, 0) -- (\adjlength, 0)
\foreach \x in {1,..., \adjside} {-- ++(0, \adjlength) -- ++(\adjlength, 0) } -- cycle;
}

\newcommand{\adjblock}[4]{
\pgfmathtruncatemacro{\adjside}{#1 - 1}
\pgfmathtruncatemacro{\adjend}{#1 + #2}
\fill[fill = #4, line width = 1pt] (0, #1\adjlength) -- (\adjside\adjlength, #1\adjlength)
\foreach \x in {1,..., #2} {-- ++(\adjlength, 0) -- ++(0, \adjlength) } -- ++(0, #3\adjlength)  -- ++(-#2\adjlength, 0) -- ++(0, -#3\adjlength) -- (0, \adjend\adjlength) -- cycle;
\fill[fill = #4, line width = 1pt] (#1\adjlength, 0) -- (#1\adjlength, \adjside\adjlength)
\foreach \x in {1,..., #2} {-- ++(0, \adjlength) -- ++(\adjlength, 0) } -- ++(#3\adjlength, 0)  -- ++(0, -#2\adjlength) -- ++(-#3\adjlength, 0) -- (\adjend\adjlength, 0) -- cycle;
}

\newcommand{\adjshort}[4]{
\pgfmathtruncatemacro{\adjside}{#1 - 1}
\pgfmathtruncatemacro{\adjend}{#1 + #2 + #3}
\fill[fill = #4, line width = 1pt] (0, #1\adjlength) -- (\adjside\adjlength, #1\adjlength)
\ifnum#2>0
\foreach \x in {1,..., #2} {-- ++(\adjlength, 0) -- ++(0, \adjlength) }
\fi
-- ++(0, #3\adjlength)  -- (0, \adjend\adjlength) -- cycle;
\fill[fill = #4, line width = 1pt] (#1\adjlength, 0) -- (#1\adjlength, \adjside\adjlength)
\ifnum#2>0
\foreach \x in {1,..., #2} {-- ++(0, \adjlength) -- ++(\adjlength, 0) }
\fi
-- ++(#3\adjlength, 0)  -- (\adjend\adjlength, 0) -- cycle;
}

\newcommand{\adjgrid}[3]{
\pgfmathtruncatemacro{\adjside}{#1 - 1}
\pgfmathtruncatemacro{\adjend}{#2 + 1}
\foreach \x in {#1,..., #2} {
\draw[draw = #3, line width = 0.5pt] (0, \x\adjlength) -- (\adjend\adjlength, \x\adjlength);
}
\foreach \x in {1,..., #1} {
\draw[draw = #3, line width = 0.5pt] (#1\adjlength, \x\adjlength) -- (\adjend\adjlength, \x\adjlength);
}
\foreach \x in {#1,..., #2} {
\draw[draw = #3, line width = 0.5pt] (\x\adjlength, 0) -- (\x\adjlength, \adjend\adjlength);
}
\foreach \x in {1,..., #1} {
\draw[draw = #3, line width = 0.5pt] (\x\adjlength, #1\adjlength) -- (\x\adjlength, \adjend\adjlength);
}
\draw[draw = #3, line width = 1pt] (0, #1\adjlength) -- (\adjside\adjlength, #1\adjlength)
\foreach \x in {#1,..., #2} {-- ++(\adjlength, 0) -- ++(0, \adjlength) } -- (0, \adjend\adjlength) -- cycle;
\draw[draw = #3, line width = 1pt] (#1\adjlength, 0) -- (#1\adjlength, \adjside\adjlength)
\foreach \x in {#1,..., #2} {-- ++(0, \adjlength) -- ++(\adjlength, 0) } -- (\adjend\adjlength, 0) -- cycle;
}

\xspaceaddexceptions{$}

\makeatletter
\tikzset{circle split part fill/.style  args={#1,#2}{%
 alias=tmp@name, 
  postaction={%
    insert path={
     \pgfextra{%
     \pgfpointdiff{\pgfpointanchor{\pgf@node@name}{center}}%
                  {\pgfpointanchor{\pgf@node@name}{east}}%
     \pgfmathsetmacro\insiderad{\pgf@x}
      \fill[#1] (\pgf@node@name.base) ([xshift=-\pgflinewidth]\pgf@node@name.east) arc
                          (0:360:\insiderad-\pgflinewidth)--cycle;
      \fill[#2] (\pgf@node@name.base) ([xshift=\pgflinewidth]\pgf@node@name.west)  arc
                           (180:360:\insiderad-\pgflinewidth)--cycle;            
         }}}}}  
\makeatother

\newcommand{\fullproof}{1}

\begin{document}

\title{Discovering Nested Communities}

\author{Nikolaj Tatti \and Aristides Gionis}
\institute{Helsinki Institute for Information Technology \\  
Department of Information and Computer Science\\
Aalto University\\
\url{{nikolaj.tatti,aristides.gionis}@aalto.fi}}

\maketitle

\begin{abstract}
Finding communities in graphs is one of the most well-studied problems in data mining and social-network analysis.
In many real applications, the underlying graph does not have a clear community structure. 
In those cases, selecting a single community turns out to be a fairly ill-posed problem, 
as the optimization criterion has to make a difficult choice between selecting a tight but small community or a more inclusive but sparser community.

In order to avoid the problem of selecting only a single community we propose
discovering a sequence of nested communities. 
More formally, given a graph and a starting set, our goal is to discover a sequence of communities all containing the starting set, and each community forming a denser subgraph than the next.
Discovering an optimal sequence of communities is a complex optimization problem, 
and hence we divide it into two subproblems: 1) discover the optimal
sequence for a fixed order of graph vertices, a subproblem that we can solve
efficiently, and 2) find a good order.
We employ a simple heuristic for discovering an order and we provide empirical and theoretical evidence that our order is good.
\end{abstract}

\section{Introduction}
\label{sec:intro}

Discovering communities, tightly connected subgraphs, is one of the most
well-studied problems in the field of graph mining.
Given some optimization criterion, discovering a community is a computationally
challending task, typically \np-hard. 
Additionally, as pointed out by Leskovec et al.~\cite{DBLP:conf/www/LeskovecLDM08}, in many real applications 
the underlying graph does not have a clear community structure.
Such cases make the community-finding problem inherently ill-posed, 
as the optimization criterion has to make a difficult, and eventually
arbitrary, choice between selecting a tight but small community
or a more inclusive but more sparse community. 
Moreover, the existence of a universal criterion for making such a
choice is unlikely as the balance between the size and the density of
the desired community will depend on the underlying application.

In order to avoid the problem of selecting only a single community,
we propose a problem of discovering a \emph{sequence of nested
communities}.  
More formally, given a graph $G$ and a set of source vertices $S$, our goal is
to discover a sequence of $k$ communities around $S$, such that each community
is a subset of the next one. 
The first community will consist only of $S$ while the last community
will contain the whole graph. 
Inner communities should be tighter than the outer communities. 
We express this requirement by computing the density of each community
and require that the next community should have a lower density than
the current community.
In addition, we require that each community should be as uniform as possible.
We measure uniformity by computing the variance of weights of the
edges and requiring it to be small.

Discovering a sequence of communities by optimizing the uniformity criterion is a
challenging problem. We will show that several optimization problems related to
the optimal solution are \np-hard. 
Hence, we split the problem into two subproblems. 
We can view a community sequence as a bucket order on the
vertices, each bucket consisting of vertices contained in the community and not
contained in the previous community.  
Our first subproblem is to discover a total order on the vertices
respecting the optimal bucket order.  
The second subproblem is to discover the optimal sequence of
communities, given an order on the graph vertices.  
Fortunately, this subproblem 
can be formulated as a standard sequence-segmentation problem, and
thus, it can be solved in polynomial time.
In particular, we can solve this problem optimally in quadratic time
or we can find an approximate solution in nearly-linear time.  
Discovering the order is more difficult as this is a complex
combinatorial problem.  
We propose a simple ordering technique used for discovering dense
subgraphs: pick iteratively a vertex with the lowest degree, and
remove it from the graph. 
We provide theoretical evidence implying that this is a good order and we also
show experimentally that this order outperforms several baselines.

The rest of the paper is organized as follows. We introduce preliminary
notation in Section~\ref{sec:prel} and formalize our optimization problem in
Section~\ref{sec:nested}. In section~\ref{sec:chain} we develop our discovery
algorithm and point out theoretical properties of our approach.
Section~\ref{sec:related} is devoted to related work and Section~\ref{sec:exps}
is devoted to experimental evaluation. We conclude our paper with a
short conclusion in Section~\ref{sec:concl}.

\section{Preliminaries}
\label{sec:prel}

We consider a weighted undirected graph $G = (V, E,\wght{})$ over a
set of vertices $V$ and edges $E \subseteq {\choosetwo{V}}$. 
We use the notation $\choosetwo{V}$ to denote the set of
unordered pairs of distinct vertices from $V$.
The function $\wght{} : E\rightarrow\reals$
assigns a weight $\wght{e}$ to each edge $e \in E$.
Also, given a subset of vertices $V'\subseteq V$ we denote by $E(V')$
the set of edges in the {\em induced} subgraph of $G$ defined by~$V'$.

The definitions and algorithms in this paper rely on a notion of 
{\em  edge density}, 
which is defined not only over subsets of vertices, 
but also over arbitrary {\em pairs} of subsets of vertices. 
Even though it is conceptually simple,
our edge-density definition requires slightly complex notation for
determining the set of potential edges to be used as a denominator in
the density ratio. 
To simplify our presentation we use the notation described below. 

Given the graph $G = (V, E,\wght{})$, we consider its 
{\em completed} representation  $\full{G}=(V,\full{E},\fullwght{})$, 
where $\full{E}=\choosetwo{V}$, 
and where $\fullwght{}$ is an extension of $\wght{}$, so that 
$\fullwght{e}=\wght{e}$ if $e\in E$, 
and $\fullwght{e}=0$ if $e\not\in E$.
In other words, $\full{G}$ can be seen as a complete graph, where all
non-edges of $G$ become zero-weight edges in $\full{G}$.
We note again that we use the completed graph representation only to
simplify our notation; in our implementation there is no need to store
the zero-weight edges. 

Now consider the completed representation
$\full{G}=(V,\full{E},\fullwght{})$ of a graph $G$, and let
$F\subseteq\full{E}$ be a non-empty subset of edges. 
We define the {\em weight} and {\em density} of $F$ as
\[
\wght{F} = \sum_{e \in F} \wght{e} 
\,\text{ and }\, 
\dnst{F} = \frac{\wght{F}}{\abs{F}}.
\]

Consider now two subsets of vertices $S, T \subseteq V$.
We define the set of {\em cross edges} from $S$ to $T$ as 
$\nbhd{S, T} = \set{(x, y) \in E \mid x \in S, y \in T}$.
It is important to note that we do not impose any constraint on the
sets $S$ and $T$; they may overlap in an arbitrary way. 
For instance, if the sets $S$ and $T$ are disjoint the edges in
$\nbhd{S,T}$ are the {\em cut} edges from $S$ to $T$, while if
$S\subseteq T$ the edge set $\nbhd{S,T}$ contains, among others, 
all the edges within~$S$.

Finally, we write $\wght{S, T}$ as a shorthand of $\wght{\nbhd{S, T}}$
and we write $\dnst{S, T}$ as a shorthand of $\dnst{\nbhd{S, T}}$.

\section{Nested Communities}
\label{sec:nested}

As we discussed in the introduction, our goal is to find the optimal
sequence of nested communities, with respect to a set of source
vertices of the input graph. 
We denote this set of source vertices by~$S$.
For conceptual simplicity, one may think of $S$ as a singleton set,
that is, identifying the sequence of nested communities for a single
vertex. 
However, all our problem definitions, algorithms, and proofs, hold for the
general case of $S$ being any subset of~$V$.

Our objective is to find $k$ nested communities, where the
parameter $k$ is part of the problem input.
Given a set of source vertices $S$, we represent a sequence of nested
communities with respect to $S$, by the sequence of vertex sets
$S= V_0 \subseteq V_1 \subseteq \cdots \subseteq V_k = V$.

Intuitively, the inner sets of the nested-community sequence are
expected to be more strongly related to the source set~$S$.
This type of relatedness is expressed by the notion of density. 
So, $V_1$ is the densest community that contains $S$, 
$V_2$ is the second densest community, and in general, we require that
the density of $V_i$ should decrease as $i$ increases.

Considering the requirement of monotonically decreasing density in
isolation is not sufficient to determine in a well-defined manner a
desirable sequence of nested communities.
Indeed, given a graph $G$, a set of source vertices $S$, and integer
$k$, there is a potentially exponential number of ways to partition
the set of vertices of the graph into a sequence of nested communities
$V_0,\ldots,V_k$.

The main question we are facing is to decide where exactly to draw the
boundary between each pair of communities $V_i$ and $V_{i+1}$. 
To answer this question, we follow an approach inspired by
{\em segmentation problems}.  
In particular, our approach is as follows:  
consider the set of vertices $D_{i+1}=V_{i+1}\setminus V_i$ that need
to be added to the community $V_i$ in order to form
community~$V_{i+1}$.
Consider also the set of edges  $E_{i+1} = E(V_{i + 1}) \setminus E(V_i)$, defined as the
additional edges brought in by extending the community  $V_i$ to the
community $V_{i+1}$.  
We can then define the density of the set of edges $E_{i+1}$.
To capture the intuition that the set $D_{i+1}$ should form a coherent
extension to $V_i$ we require that the density of $E_{i+1}$ is as
{\em uniform} as possible.

The notion of uniformity for a set of edges, among many ways, can be
expressed as a sum of square of difference of the weight of each edge
from the average weight of the set.
We thus have the following definition.

\begin{definition}
Given a set of edges $F \subseteq E$, we define the
{\em density-uniformity score} as
\[
	\score{F} = \sum_{e \in F} \pr{\wght{e} - \dnst{F}}^2.
\]
\end{definition}

Our goal is then to find a sequence of nested communities so that the
successive segments of added edges are as uniform as possible with
respect to their density.
Formulating this objective as an optimization problem not only gives
meaningful semantics  to the nested community detection problem, but
it also makes the problem well-defined. 
Motivated by the discussion above, our main problem definition is given below.

\begin{problem}
\label{prb:chain}
Given a weighted input graph $G = (V, E, \wght{})$, a set of
source vertices $S \subset V$, and an integer $k$, 
find the sequence of nested communities  
$\fm{V} = \{ S = V_0 \subseteq V_1 \subseteq \cdots \subseteq V_k = V\}$
that minimizes the density-uniformity score
\[
\score{\fm{V}} = \sum_{i = 1}^k \score{E(V_i) \setminus E(V_{i - 1})},
\]
subject to the constraint  $\dnst{V_i} < \dnst{V_{i - 1}}$ for $i = 2, \ldots, k$.
\end{problem}

\section{An Algorithm for Discovering Nested Communities}\label{sec:chain}

In this section we present our algorithm for discovering nested communities.
We begin by demonstrating a necessary condition for the optimal solution based
on dense subgraphs. Discovering such subgraphs turns out to be computationally
intractable. We then split the original problem into two subproblems:
discovering community sequence for a fixed order of vertices, a problem which
we can solve efficiently, and discovering such an order. We provide a simple heuristic
for discovering an order, and provide theoretical evidence
that this order is good.

\subsection{Nested Communities and Dense Subgraphs}
\label{section:dense-and-sparse}

We start our discussion by demonstrating a connection of the problem
of finding the optimal sequence of nested communities, i.e., solving
Problem~\ref{prb:chain}, with problems related to finding dense
subgraphs of a given graph. 

To establish this connection, consider a triple of communities 
$V_{i-1} \subseteq V_i \subseteq V_{i+1}$ in an {\em optimal solution}
to Problem~\ref{prb:chain}.
Consider the two corresponding segments
$D_{i+1} = V_{i+1} \setminus V_i$ and
$D_i = V_i \setminus V_{i-1}$.
Consider also any two subsets of those segments,
$X \subseteq D_{i+1}$ and  $Y \subseteq D_i$, 
that is, $X$ is a subset of the outer segment, 
while $Y$ is a subset of the inner segment, see Figure~\ref{fig:prop1:a} for a visualization.
As we will show shortly, adding the outer subset $X$ in the community 
$V_i$ leads to a situation where the density of the subset $X$ with
respect to the overall community $V_i$ is no better than the density
of the subset $Y$ with respect to the community $V_i$.
Otherwise, either adding $X$ to $V_i$ (see Figure~\ref{fig:prop1:b}) or removing $Y$ from $V_i$
(see Figure~\ref{fig:prop1:c}) lead to a better solution.
This follows from the fact that we require that the
densities of the nested communities in any feasible solution of
Problem~\ref{prb:chain} decrease monotonically.

\begin{figure}[t]
\begin{center}

\subfigure[\label{fig:prop1:a}Original community]{
\begin{minipage}[b]{3.5cm}
\begin{center}
\begin{tikzpicture}[very thick, scale = 0.8]
\draw[yafcolor5, fill = yafcolor5!10] (0, 0) circle (1.7);
\draw[yafcolor5, fill = yafcolor5!20] (-45:1) .. controls +(45:0.3) and +(-135:0.3) .. (1.4, -0.28) arc (-45:45:0.4) .. controls +(135:0.3) and +(-45:0.3) .. (45:1);
\draw[yafcolor3, fill = yafcolor3!30] (0, 0) circle (1);
\draw[yafcolor2, fill = yafcolor2!30] (-0.1, 0.3) circle (0.3);
\draw[yafcolor3, fill = yafcolor3!20] (-30:1) .. controls +(-135:0.3) and +(-45:0.3) .. (0.28, -0.4) arc (45:135:0.4) .. controls +(-135:0.3) and +(-45:0.3) .. (-160:1) arc (-160:-30:1);
\node[] at (-0.1, 0.3) {$S$};
\node[] at (0, -0.6) {$Y$};
\node[] at (1.2, 0) {$X$};
\node[] at (0.6, 0) {$D_1$};
\node[] at (-0.95, 0.95) {$D_2$};
\end{tikzpicture}
\end{center}
\end{minipage}}\hfill
\subfigure[\label{fig:prop1:b}Adding $X$]{
\begin{minipage}[b]{3cm}
\begin{tikzpicture}[very thick, scale = 0.8]
\draw[yafcolor5, fill = yafcolor5!10] (0, 0) circle (1.7);
\draw[yafcolor3, fill = yafcolor3!30] (-45:1) .. controls +(45:0.3) and +(-135:0.3) .. (1.4, -0.28) arc (-45:45:0.4) .. controls +(135:0.3) and +(-45:0.3) .. (45:1)
arc (45:315:1) -- cycle;
\draw[yafcolor2, fill = yafcolor2!30] (-0.1, 0.3) circle (0.3);
\node[] at (-0.1, 0.3) {$S$};
\node[] at (0.6, 0) {$D_1$};
\node[] at (-0.95, 0.95) {$D_2$};
\end{tikzpicture}
\end{minipage}}\hfill
\subfigure[\label{fig:prop1:c}Removing $Y$]{
\begin{minipage}[b]{3cm}
\begin{tikzpicture}[very thick, scale = 0.8]
\draw[yafcolor5, fill = yafcolor5!10] (0, 0) circle (1.7);
\draw[yafcolor3, fill = yafcolor3!20] (-30:1) .. controls +(-120:0.3) and +(-45:0.3) .. (0.28, -0.4) arc (45:135:0.4) .. controls +(-135:0.3) and +(-60:0.3) .. (-160:1) arc (200:-30:1) -- cycle;
\draw[yafcolor2, fill = yafcolor2!30] (-0.1, 0.3) circle (0.3);
\node[] at (-0.1, 0.3) {$S$};
\node[] at (0.6, 0) {$D_1$};
\node[] at (-0.95, 0.95) {$D_2$};
\end{tikzpicture}
\end{minipage}}
\end{center}
\caption{Communities related to Proposition~\ref{proposition:density}. If $\dnst{X, X \cup D_1} > \dnst{Y, D_1}$, then either adding $X$ to $D_1$ or removing $Y$ from $D_1$ will yield a better score.}
\end{figure}
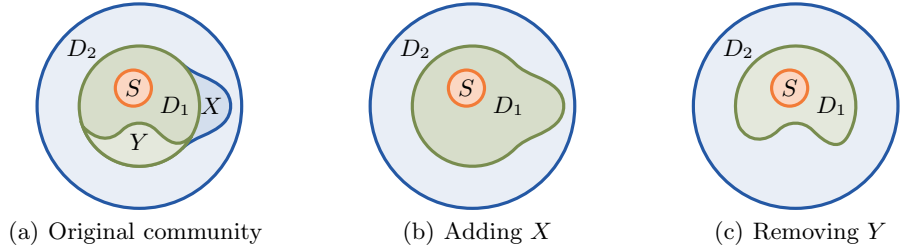

Before proceeding to discussing the implications of this observation,
we first give a formal statement and its proof.

\begin{proposition}
\label{proposition:density}
Consider a graph $G = (V, E, \wght{})$, a set of
source vertices $S \subseteq V$, and an integer $k$. 
Let $\fm{V} = \pr{ S = V_0 \subseteq V_1 \subseteq \cdots \subseteq
  V_k = V}$ 
be the optimal sequence of nested communities, that is, 
a solution to Problem~\ref{prb:chain}.
Fix $i$ such that $1 \leq i \leq k - 1$ and let $X \subseteq V_{i + 1} \setminus V_i$
and $Y  \subseteq V_i \setminus V_{i - 1}$. 
Then 
\[
\dnst{X, X \cup V_i} \leq \dnst{Y, V_i}.
\]
\end{proposition}

For the proof of the proposition we require the
following lemma, which states that the mean square error
of a set of numbers from a single point, increases with the distance of that
point from the mean of the numbers. 
The lemma can be derived by simple algebraic manipulations, and its proof is omitted. 

\begin{lemma}
\label{lem:centroid}
Let $\enlst{w_1}{w_N}$ and $\enlst{x_1}{x_N}$ be two sets of real numbers.
Let $W = \sum_{i = 1}^N w_i$ and $\mu = \frac{1}{W} \sum_{i = 1}^N w_ix_i$. 
For any real number $d$ it is 
\[
	\sum_{i = 1}^N w_i(x_i - d)^2 = \sum_{i = 1}^N w_i(x_i - \mu)^2 + W(d - \mu)^2.
\]
\end{lemma}

We are now ready to prove the proposition.

\begin{proof}[Proposition \ref{proposition:density}]
Let $C_1 = E(V_{i + 1}) \setminus E(V_i)$ and $C_2 = E(V_i) \setminus E(V_{i - 1})$.
Let us break $C_1$ into two parts,
$D_{11} = \nbhd{X, X \cup V_i}$ and $D_{12} = C_1 \setminus D_{11}$.
Similarly, let us break $C_2$ into two parts,
$D_{21} = \nbhd{Y, V_i}$ and $D_{22} = C_2 \setminus D_{21}$.
Define the centroids $\mu_{ij} = \dnst{D_{ij}}$ and $\lambda_i = \dnst{C_i}$.
Lemma~\ref{lem:centroid} now implies that
\begin{eqnarray*}
	s = \score{C_1} + \score{C_2} & = & \text{const} + \abs{D_{11}}(\mu_{11} - \lambda_1)^2 + \abs{D_{21}}(\mu_{21} - \lambda_2)^2,\\
	s_1 = \score{C_1 \cup D_{21}} + \score{D_{22}} & = & \text{const} + \abs{D_{11}}(\mu_{11} - \lambda_1)^2 + \abs{D_{21}}(\mu_{21} - \lambda_1)^2,\\
	s_2 = \score{D_{12}} + \score{C_1 \cup D_{11}} & = & \text{const} + \abs{D_{11}}(\mu_{11} - \lambda_2)^2 + \abs{D_{21}}(\mu_{21} - \lambda_2)^2,\\
\end{eqnarray*}
where const is equal to
\[
	\sum_{i = 1}^2 \score{D_{i1}} + \score{D_{i2}} + \abs{D_{i2}}(\mu_{i2} - \lambda_i)^2\quad.
\]
Since $\fm{V}$ is optimal we must have $s \leq s_1$ and $s \leq s_2$. Otherwise, we can obtain
a better segmentation by attaching $X$ to $V_i$ or deleting $Y$ from $V_i$.
This implies that $\abs{\mu_{21} - \lambda_2} \leq \abs{\mu_{21} - \lambda_1}$ and
$\abs{\mu_{11} - \lambda_1} \leq \abs{\mu_{11} - \lambda_2}$. 
Since $\lambda_2 \geq \lambda_1$, this implies that $\mu_{21} \geq (\lambda_1 + \lambda_2)/2$
and $\mu_{11} \leq (\lambda_1 + \lambda_2)/2$, which implies  $\mu_{11} \leq \mu_{21}$.
This completes the proof.  \qed
\end{proof}

Proposition~\ref{proposition:density} implies that in an optimal
solution the graph vertices can be {\em ordered} in such a way so that
subgraph density, as specified by the proposition, decreases along
this order. 
This observation motivates the following {\em greedy} algorithm for
solving the problem of discovering nested communities:

\medskip
\noindent
{\bf Algorithm outline:} Greedy--add--densest--subgraph
\begin{enumerate}
\item
Start with $S$, the set of source vertices.
\item
Given the current set $S$, find a subset of vertices $T$ that maximize
$\dnst{T, S \cup T}$.
\item
Set $S\leftarrow S\cup T$, and repeat the previous step until
the set $S$ includes all the vertices of the graph.
\item
Consider the vertices in the order discovered by the previous
process. 
Find the optimal sequence of $k$ nested communities that respects this order.
\end{enumerate}

One potential problem with the above greedy approach is that the
subroutine that is called iteratively in step 2, is an \np-hard problem.
This is formalized below as problem \textsc{DenseSuperset}.
\if\fullproof1
The proof of Proposition~\ref{proposition:np-dense} is given in Section~\ref{section:np-hardness}.
\fi

\begin{problem}[\textsc{DenseSuperset}]
Given a weighted graph $G = (V, E, \wght{})$ and a subset of vertices $S
\subseteq V$, 
find a subset of vertices $T$ maximizing  $\dnst{T, S \cup T}$.
\end{problem}

\begin{proposition}
\label{proposition:np-dense}
The \textsc{DenseSuperset} problem is \np-hard.
\end{proposition}

\if\fullproof0
\begin{proof}[Sketch]
Due to space constraints we will only sketch the proof. The complete proof
is available in Appendix.\!\footnote{For the appendix, see \url{http://users.ics.aalto.fi/~ntatti/}} We will reduce \textsc{Clique} to \textsc{DenseSuperset}.
Given a graph $G$, we add a vertex $s$ and connect it to each vertex with a weight of
$\alpha = 1 - \frac{1}{2 \abs{V}^2}$. Let $k < n < m$. It follows that
\[
	\frac{\choosetwo{n}  + \alpha n}{{n \choose 2} + n} > \frac{\choosetwo{k}  + \alpha k}{{k \choose 2} + k} \quad\text{and}\quad
	\frac{\choosetwo{n}  + \alpha n}{{n \choose 2} + n} > \frac{\choosetwo{m}  + \alpha m - 1}{{k \choose 2} + m}\quad.
\]
The left-hand side term in the first equation is the density of $n$-clique while the
the right-hand side term bounds the density of a graph with $k$ vertices. The right-hand side term in the second equation upper bounds
the density of a non-clique with $m$ vertices. Consequently, the
largest clique, say $X$, in $G$ will also have the largest density $\dnst{X, X \cup s}$, which is a sufficient to prove the result.\qed
\end{proof}
\fi

Similarly, one can think of solving the problem by working on the
opposite direction, that is, start with the whole vertex set $V$
and ``peel off'' the set $V$ by removing the sparsest subgraph, until
left with the set of source vertices $S$.
The corresponding algorithm will be the following.

\medskip
\noindent
{\bf Algorithm outline:} Greedy--remove--sparsest--subgraph
\begin{enumerate}
\item
Start with $V$, the vertex set of $G$.
\item
Given a current set $V$,
find a subset of vertices $T$ that does not include the source vertex set
$S$ and minimizes the density $\dnst{T, V}$.
\item
Set $V\leftarrow V\setminus T$, and repeat the previous step until
left only with the set of source vertices $S$.
\item
Consider the vertices in the order removed by the previous process. 
Find the optimal sequence of $k$ nested communities that respects this order.
\end{enumerate}

Not surprisingly, the problem of finding the sparsest subgraph, which
corresponds to step 2 of the above process is \np-hard.
\if\fullproof1
The proof is given again in Section~\ref{section:np-hardness}.
\fi

\begin{problem}[\textsc{SparseNbhd}]
Given a weighted graph $G = (V, E,\wght{})$
find a set of vertices $T$ minimizing $\dnst{T, V}$.
\end{problem}

\begin{proposition}
\label{proposition:np-sparse}
The \textsc{SparseNbhd} problem is \np-complete.
\end{proposition}

\if\fullproof0
\begin{proof}[Sketch]
Due to space constraints we will only sketch the proof. The complete proof
is available in Appendix. We will reduce \textsc{Clique} to \textsc{SparseNbhd}.
Assume that we are given a graph $G$ with $l$ nodes. We extend the graph by adding two vertices $s$ and $t$ with an edge of such high weight
that neither $s$ or $t$ will appear in the optimal solution. We then add an edge
from $s$ to each vertex $v$ in $G$ with a weight of
$p - \deg\pr{v}$, where $p = (l + 1) - \frac 12 (l + 1)^{-2}$.
This will make the weighted degree of all vertices in $G$ equal
so a dense subgraph $X$ will have a low density $\dnst{X, X \cup \set{s, t}}$.
Let $k < n < m$. Then a straightforward calculation reveals that
\[
	\frac{pn  - \choosetwo{n}}{(l + 1)n - \choosetwo{n}} < \frac{pk  - \choosetwo{k}}{(l + 1)k - \choosetwo{k}}
	\quad\text{and}\quad
	\frac{pn  - \choosetwo{n}}{(l + 1)n - \choosetwo{n}} < \frac{pm  - \choosetwo{m} + 1}{(l + 1)m - \choosetwo{m}}\quad.
\]
The left-hand side term in the first equation is the density of $n$-clique while the
right-hand side term bounds the density of a graph with $k$ vertices. The right-hand side term in the second equation lower bounds
the density of
a non-clique with $m$ vertices.
Consequently,
the largest clique, say $X$, in $G$ will also have the lowest density $\dnst{X, X \cup \set{s, t}}$, which is a sufficient to prove the result.\qed
\end{proof}
\fi

\subsection{Algorithm for Discovering Nested Communities}
\label{section:algorithms}

Armed with intuition from the previous section, we
now proceed to discuss the proposed algorithm.
The underlying principle of both of the greedy algorithms described
above is to consider the vertices of the graph in a specific order
and then find a sequence of nested communities that respects this
order.
In one case, 
the order of graph vertices is obtained by starting from
$S$ and iteratively adding the densest subgraph, 
while in the other case, 
the order is obtained by starting from the full vertex set $V$ and
iteratively removing the sparsest subgraph. 

Our algorithm is an instantiation of this general principle. 
We specify in detail 
($i$) how to obtain an order of the graph vertices, and
($ii$) how to find a sequence of nested communities that respects a
given order. 

We start our discussion from the second task, i.e.,
finding the sequence of nested communities given an order.
As it turns out, this problem is an instance of sequence
segmentation problems.
We define this problem below, which is a refinement of Problem~\ref{prb:chain}.

\begin{problem}[Sequence of nested communities from a given order]
\label{prb:chainorder}
Given a graph $G = (V, E, \wght{})$ \emph{with ordered vertices},  
a set of source vertices $S = \enset{v_1}{v_s} \subset V$, and an
integer $k$, 
find a monotonically increasing sequence of $k + 1$ integers 
$b = \enpr{b_0 = s}{b_k = \abs{V}}$ such that 
\[
\fm{V} = \pr{ S = V_0 \subseteq V_1 \subseteq \cdots \subseteq V_k = V},\quad\text{where}\quad V_k = \enset{v_1}{v_{b_i}},
\]
minimizes the density-uniformity score $\score{\fm{V}}$ and
satisfies the monotonicity constraint
$\dnst{V_i} < \dnst{V_{i - 1}}$ for $i = 1, \ldots, k$. 
\end{problem}

It is quite easy to see that Problem~\ref{prb:chainorder} can be cast
as a segmentation problem. 
Typical segmentation problems can be solved optimally using
dynamic programming, as shown by Bellman~\cite{Bellman}.
The most interesting aspect of Problem~\ref{prb:chainorder}, 
seen as segmentation problem, is the monotonicity constraint 
$\dnst{V_i} < \dnst{V_{i - 1}}$, for $i = 1, \ldots, k$. 
That is, not only we ask to segment the ordered sequence of vertices
so that we minimize the density variance on the segments,
but we also require that the density scores of each segment decrease
monotonically. 
The situation can be abstracted to the monotonic segmentation problem
stated below. 

\begin{problem}[Monotonic segmentation]
\label{prb:monotone}
Let $\enlst{a_1}{a_n}$ and $\enlst{x_1}{x_n}$ be two sequences of real numbers.
Given an integer $k$, find $k + 1$ indices $b_0 = 1, \ldots, b_k = n + 1$ minimizing 
\[
	\sum_{j = 1}^n \sum_{i = b_{j - 1}}^{b_j - 1} a_i(x_i - \mu_j)^2, 
\]
where $\mu_j$ is the weighted centroid of $j$-th segment such that $\mu_{j} < \mu_{j - 1}$.
\end{problem}

In order to express Problem~\ref{prb:chainorder} with Problem~\ref{prb:monotone}, consider
a group of edges, $P_i = \nbhd{v_i, \enset{v_1}{v_{i - 1}}}$ for each vertex $v_i \in V \setminus S$.
If we set $a_i = \abs{P_{i + \abs{S}}}$ and $x_i = \dnst{P_{i + \abs{S}}}$,
we can apply Lemma~\ref{lem:centroid} and show that the score of community sequence
is equal to the variance minimized by Problem~\ref{prb:monotone}, plus a constant.
In fact, this constant is the sum of the variances within each $P_i$.

Similarly to the unconstrained segmentation problem, the monotonic
segmentation problem can be solved {\em optimally}. 
The idea is to use as preprocessing step the classic ``pool of
adjacent violators'' algorithm (PAV)~\cite{PAV}, 
which merges points until there are no monotonicity violations, 
and then apply the classic dynamic-programming algorithm on the
resulting sequence of merged points. This algorithm runs in $O(\abs{V})$ time.
By definition the merged points do not contain any monotonicity
violations, and thus, the resulting segmentation respects the
monotonicity constraint, as well. 
As shown by Haiminen et al.~\cite{unimodal}, this two-phase algorithm 
gives the optimal $k$ segmentation under the monotonicity constraints.
As a result of the optimality of the monotonic segmentation problem, 
Problem~\ref{prb:chainorder} can be solved optimally.

\smallskip
We next proceed to discuss the first component of the algorithm, 
namely, how to obtain an order of the graph vertices. 
Recall that, according to the principles discussed in the previous
section, we can either start from $S$ and iteratively add dense
subgraphs,
or start from $V$ and remove sparse subgraphs. 
We follow the latter approach. 
In order to overcome the  \np-hard problem of finding the sparsest
subgraph and in order to obtain a total order, 
we use the heuristic of iteratively removing the sparsest subgraph of
size one, namely, a single vertex.
The sparsest one-vertex subgraph is simply the vertex with the
smallest weighted degree. 
Thus, overall, we obtain the simple algorithm \sortnodes, whose
pseudocode is given as Algorithm~\ref{algorithm:sort}.

As an interesting side remark, we note that the algorithm \sortnodes
is encountered in the context of finding subgraphs with the highest average degree. 
In particular, it is known that the densest subgraph obtained by the
algorithm during the process of iteratively removing the
smallest-degree vertex is a factor-2 approximation to the optimally
densest subgraph in the graph~\cite{Charikar}.

The natural question to ask is how good is the order produced by algorithm
\sortnodes?  As we will demonstrate shortly, it turns out that the order is
quite good.
First, we note that the optimal solution obtained for Problem~\ref{prb:chainorder},
satisfies an analogous structural property, with respect to subgraph
densities, as the optimal solution for Problem~\ref{prb:chain},
We omit the proof of the following proposition as it is similar to the one of
Proposition~\ref{proposition:density}.

\begin{proposition}
\label{proposition:optimality-order}
Consider a graph $G = (V, E, \wght{})$ with ordered vertices, a set of
source vertices $S \subset V$, and an integer $k$. 
Let $\fm{V} = \pr{ S = V_0 \subseteq V_1 \subseteq \cdots \subseteq
  V_k = V}$ 
be the optimal sequence of nested communities with respect to  the order, that is, 
a solution to Problem~\ref{prb:chain}.
Fix $i$ such that $1 \leq i \leq k - 1$ and let $b = \abs{V_i}$.
Let $X \subseteq V_{i + 1} \setminus V_i$
and $Y  \subseteq V_i \setminus V_{i - 1}$ such that $X = \enset{v_{b + 1}}{v_{b + \abs{X}}}$
and $Y = \enset{v_{b - \abs{Y} + 1}}{v_{b}}$. 
Then $\dnst{X, X \cup V_i} \leq \dnst{Y, V_i}$.
\end{proposition}

The only difference between Proposition~\ref{proposition:density} and
Proposition~\ref{proposition:optimality-order} is that in
Proposition~\ref{proposition:optimality-order} we require additionally that
$V_{i + 1}$ starts with $X$ and $V_i$ ends with $Y$ with respect to  the order. We want
this condition to be redundant, otherwise the given order is suboptimal.  For
example, consider the adjacency matrix of $G$ given in
Figure~\ref{fig:viz:b}.  
The given segmentation is optimal with respect to  the given order. 
However if we rearrange the
vertices in $D_1$ and $D_2$, given in Figure~\ref{fig:viz:c}, then the same
segmentation is no longer optimal as $X$ and $Y$ violate Proposition~\ref{proposition:optimality-order}.
The additional condition in  Proposition~\ref{proposition:optimality-order} becomes redundant if $V_i$ ends with the sparsest subset
while $V_{i + 1}$ starts with densest subset.
We will show that the algorithm \sortnodes produces an order that satisfies this
property {\em approximately}.
The exact formulation of our claim is given as
Propositions~\ref{proposition:left} and~\ref{proposition:right}.

\begin{figure}[b]
\begin{center}
\hfill
\subfigure[\label{fig:viz:b}Original order]{
\begin{tikzpicture}[yscale=-1]
\adjbox{15}{yafcolor5!50!white};
\adjbox{11}{yafcolor3!70!white};
\adjbox{3}{yafcolor2!80};
\adjblock{14}{1}{0}{yafcolor5!80};
\adjblock{13}{1}{1}{yafcolor5!90};
\adjblock{9}{2}{0}{yafcolor3!90};
\adjblock{7}{1}{3}{yafcolor3!30!white};
\adjblock{5}{1}{5}{yafcolor3!30!white};
\adjgrid{3}{10}{yafcolor3};
\adjgrid{11}{14}{yafcolor5};
\adjgrid{1}{2}{yafcolor2};
\node[inner sep = 1pt] (d1l) at (-0.4, 7\adjlength) {\scriptsize$D_1$};
\node[inner sep = 1pt] (d2l) at (-0.4, 13\adjlength) {\scriptsize$D_2$};
\node[inner sep = 1pt, anchor = east] (sl) at (-0, 1.5\adjlength) {\scriptsize$S$};

\draw[rounded corners, yafcolor3, line width = 0.6pt] (0,11\adjlength) -| (d1l.south);
\draw[rounded corners, yafcolor3, line width = 0.6pt] (0,3\adjlength) -| (d1l.north);

\draw[rounded corners, yafcolor5, line width = 0.6pt] (0,11\adjlength) -| (d2l.north);
\draw[rounded corners, yafcolor5, line width = 0.6pt] (0,15\adjlength) -| (d2l.south);

\end{tikzpicture}}\hfill
\subfigure[\label{fig:viz:c}Improved order]{
\begin{tikzpicture}[yscale=-1]
\adjbox{15}{yafcolor5!50!white};
\adjbox{11}{yafcolor3!70!white};
\adjbox{3}{yafcolor2!80};
\adjblock{12}{1}{2}{yafcolor5!80};
\adjblock{11}{1}{3}{yafcolor5!90};
\adjblock{3}{2}{4}{yafcolor3!90};
\adjblock{9}{2}{0}{yafcolor3!30};
\adjgrid{3}{10}{yafcolor3};
\adjgrid{11}{14}{yafcolor5};
\adjgrid{1}{2}{yafcolor2};

\node[inner sep = 1pt] (d1l) at (-0.4, 7\adjlength) {\scriptsize$D_1$};
\node[inner sep = 1pt] (d2l) at (-0.4, 13\adjlength) {\scriptsize$D_2$};
\node[inner sep = 1pt, anchor = east] (sl) at (-0, 1.5\adjlength) {\scriptsize$S$};

\node[inner sep = 1pt, anchor = east] at (-0, 10\adjlength) {\scriptsize$Y$};
\node[inner sep = 1pt, anchor = east] at (-0, 12\adjlength) {\scriptsize$X$};

\draw[rounded corners, yafcolor3, line width = 0.6pt] (0,11\adjlength) -| (d1l.south);
\draw[rounded corners, yafcolor3, line width = 0.6pt] (0,3\adjlength) -| (d1l.north);

\draw[rounded corners, yafcolor5, line width = 0.6pt] (0,11\adjlength) -| (d2l.north);
\draw[rounded corners, yafcolor5, line width = 0.6pt] (0,15\adjlength) -| (d2l.south);

\end{tikzpicture}}
\hspace*{\fill}
\end{center}
\caption{Consequences of Proposition~\ref{proposition:optimality-order}. If we reorder the vertices in $D_1$ and $D_2$, then an optimal solution
with respect to  the order may become suboptimal with respect to the improved order.}
\end{figure}
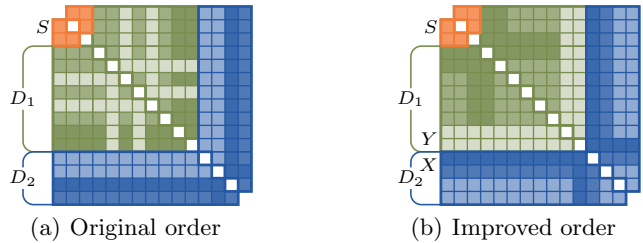

\begin{algorithm}
\label{algorithm:sort}
\caption{\sortnodes. Sort vertices of a weighted graph by iteratively removing a vertex with the least weight of adjacent edges.}
\Input{weighted graph $G = (V, E,\wght{})$, a set $S$}
\Output{order on $V$}

$W \define V \setminus S$\;
$o \define$ empty sequence\;
\While {$\abs{W} > 0$} {
	$x \define \arg \min_{x \in W} \dnst{x, W \cup S}$\;
	delete $x$ from $W$ and add $x$ to the beginning of $o$\;
}
add $S$ in an arbitrary order to the beginning of $o$\;
\Return{$o$}\;
\end{algorithm}

\begin{proposition}
\label{proposition:left}
Consider a weighted graph $G = (V, E,\wght{})$, whose vertices are
ordered by algorithm \sortnodes.  
Let $1 \leq b < c \leq \abs{V}$. 
Let $U = \enset{v_b}{v_c}$ and $W = \enset{v_1}{v_c}$. 
Let $f = \dnst{v_c, W}$.
Then $2f \leq \dnst{X, W}$ for any $X \subseteq U$.
\end{proposition}

\begin{proof}
Note that $s = \sum_{x \in X}\wght{x, W} = 2\wght{X} + \wght{X, W\setminus X} \leq 2\wght{X, W}$.
Write $m_x = \abs{\nbhd{x, W}}$. Since $v_c$ has the smallest $\dnst{v_c, W}$, we have
\[
	s =  \sum_{x \in X}m_x\dnst{x, W} \geq \dnst{v_c, W} \sum_{x \in X}m_x \geq \dnst{v_c, W} \abs{\nbhd{X, W}}\quad.
\]
Combining the inequalities and dividing by $\abs{\nbhd{X, W}}$ proves the result.\qed
\end{proof}

\begin{proposition}
\label{proposition:right}
Consider a weighted graph $G = (V, E,\wght{})$, 
whose vertices are ordered by algorithm \sortnodes.  
Let $1 \leq b < c \leq \abs{V}$. 
Let $U = \enset{v_b}{v_c}$ and $W = \enset{v_1}{v_{b - 1}}$.
Assume that there is $\alpha \geq 0$ such that for all 
$v \in U$ it is $\alpha\wght{v, W} \geq \wght{v, U}$.
Let $f = \dnst{v_b, W}$.
Then $(1 + \alpha)^2f \geq \dnst{X, X \cup W}$ for any $X \subseteq U$.
\end{proposition}

\begin{proof}
Let $A = \nbhd{X, W}$ and $B = \nbhd{X, X}$. The density of $X$ is bounded by
\[
	\dnst{X, X \cup W} = \frac{\wght{A} + \wght{B}}{\abs{A} + \abs{B}} \leq \frac{\wght{A} + \alpha\wght{A}}{\abs{A} + \abs{B}} \leq \frac{(1 + \alpha)\wght{A}}{\abs{A}} = (1 + \alpha)\dnst{A}.
\]
Select $x \in X$ with the highest $\dnst{x, W}$. Then $\dnst{A} \leq \dnst{x, W}$. 
Let us prove that $\dnst{x, W} \leq (1 + \alpha)f$.
If $v_b = x$, then we are done. Assume that $v_b \neq x$. Since $G$ is fully-connected,
\sortnodes always picks the vertex with the lowest weight. Let $Z = \enset{v_1}{x}$.
Then
$\wght{x, W} \leq \wght{x, Z} \leq \wght{v_b, Z} = \wght{v_b, W} + \wght{v_b, U} \leq (1 + \alpha)\wght{v_b, W}$.
Since, $G$ is fully-connected $\wght{y, W} = \abs{W}\dnst{y, W}$ for any $y \in U$.
Hence, dividing the inequality gives us $\dnst{x, W} \leq (1 + \alpha)f$, which proves the proposition.\qed
\end{proof}

\if\fullproof1
\subsection{Hardness of Finding Dense and Sparse Subgraphs}
\label{section:np-hardness}

In this section we prove the \np-hardness results, stated in
Section~\ref{section:dense-and-sparse}.
We start with an auxiliary lemma.

\begin{lemma}
\label{lem:ratio}
Let $x, y, a, b, c$ be real numbers.
Let $r = b + (b + x)c / (y - x)$.
If
\[
a > r \text{ and } y > x \text{ or if } a < r \text{ and } x < y, \text{ then } \frac{x + a}{x + b} > \frac{y + a + c}{y + b}\quad.
\]
Similarly, if
\[
a < r \text{ and } y > x \text{ or if } a > r \text{ and } x < y, \text{ then } \frac{x + a}{x + b} < \frac{y + a + c}{y + b}\quad.
\]

\end{lemma}

\begin{proof}
We will only prove the first case. The other 3 cases are similar.
We have $(x - y)a > (x - y)b + (b + x)c$ which is equivalent to $xy + ay + xb + ab> xy + ax + cx + by + bc + ab$.
The left-hand side is equal to $(x + a)(y + b)$ while the right hand side is equal to $(y + a + c)(x + b)$.
The lemma follows.\qed
\end{proof}

We now give the proofs of Propositions~\ref{proposition:np-dense} and~\ref{proposition:np-sparse}.

{
\renewcommand{\theproposition}{\ref{proposition:np-dense}}
\begin{proposition}
The \textsc{DenseSuperset} problem is \np-hard.
\end{proposition}
\addtocounter{proposition}{-1}
}

\begin{proof}
To prove the hardness, we will reduce \textsc{Clique}
to \textsc{DenseSuperset}. Let $G = (V, E)$ be the given graph.
Let us create a new graph $G'$ by adding one extra vertex, say $s$, to $G$ and connecting
every vertex in $G$ to $s$. We set $\wght{e}$ to be $1$ for any edge in $E$ and 
$\alpha$, which we will define later, if $e$ is adjacent to $s$.
Finally, we connect the non-connected vertices with edges of weight $0$.
We will use $G'$, $S = \set{s}$, and $\wght{}$ as inputs to \textsc{DenseSuperset}.

Our next step is to define $\alpha$ such that the maximum clique
will also have the largest density. In order to do that,
let $X$ be a clique of size $N$ in $G$. Then the weight of $X$ is equal to
\[
	\dnst{X, X \cup S} = \frac{{N \choose 2}  + \alpha N}{{N \choose 2} + N} = \frac{N - 1 + 2\alpha}{N  -1 + 2}\quad.
\]

If we have a non-clique subgraph of size $N$, then obviously its weight is
genuinely smaller than $\dnst{X, X \cup S}$.

Assume a set of vertices $Z$ with $K < N$ vertices. The weight of $Z$ is bounded by
\[
    \dnst{Z, Z \cup S} \leq \frac{{K \choose 2}  + \alpha K}{{K \choose 2} + K} = \frac{K - 1 + 2\alpha}{K - 1  + 2}\quad.
\]
We want $\dnst{X, X \cup S} > \dnst{Z, Z \cup S}$, which is guaranteed if
\begin{equation}
\label{eq:boundabove}
	\frac{N - 1 + 2\alpha}{N + 1} > \frac{K - 1 + 2\alpha}{K + 1}\quad.
\end{equation}
Since $N - 1 > K - 1$, Lemma~\ref{lem:ratio} implies that if
\[
	2\alpha < 2 + \frac{2 + N - 1}{(K - 1) - (N - 1)}0 = 2,
\]
then the inequality in Eq~\ref{eq:boundabove} is guaranteed.

Let $Y$ be a non-clique of size $M > N$ in $G$. Then the weight of $Y$ bounded by
\[
    \dnst{Y, Y \cup S} \leq \frac{{M \choose 2}  + \alpha M - 1}{{M \choose 2} + M} = \frac{M - 1 + 2\alpha - 2/M}{M - 1  + 2}\quad.
\]
We need to have $\dnst{X, X \cup S} > \dnst{Y, Y \cup S}$, which is guaranteed if
\begin{equation}
\label{eq:boundbelow}
  \frac{N - 1 + 2\alpha}{N  -1 + 2} > \frac{M - 1 + 2\alpha - 2/M}{M - 1  + 2}\quad.
\end{equation}

Since $N - 1 < M - 1$, Lemma~\ref{lem:ratio} guarantees that if
\[
	2\alpha > 2 + \frac{-2}{M}\frac{2 + N - 1}{(M - 1) - (N - 1)} = 2 - \frac{2(N + 1)}{M(M - N)},
\]
then the inequality in Eq.~\ref{eq:boundbelow} is guaranteed.
If we choose $\alpha = 1 - 0.5 / \abs{V}^2$, both inequalities in
Eqs.~\ref{eq:boundabove}--\ref{eq:boundbelow} are now guaranteed.

Let $k$ be the minimum size of the clique given as a parameter in \textsc{Clique}. Set $\beta = \frac{k - 1 + 2\alpha}{k  -1 + 2}$. 
If $G$ contains a clique of size $k$, then there is a subgraph in $G'$ with a density of $\beta$.
Assume now that $G'$ contains a subgraph, say $H$, with a density of at least $\beta$. $H$ must contain
at least $k$ vertices, otherwise bound in Eq.~\ref{eq:boundabove} is violated. $H$ must be a clique, otherwise
bound in Eq.~\ref{eq:boundbelow} is violated. Consequently, $G$ has a clique of size $k$ if and only if
$G'$ has a subgraph of density at least $\beta$. The reduction is polynomial.
This concludes the proof.\qed
\end{proof}

{
\renewcommand{\theproposition}{\ref{proposition:np-sparse}}
\begin{proposition}
The \textsc{SparseNbhd} problem is \np-hard.
\end{proposition}
\addtocounter{proposition}{-1}
}

\begin{proof}
To prove the hardness, we will reduce \textsc{Clique}
to \textsc{SparseNbhd}. Let $G = (V, E)$ be the given graph.
We will define $G' = (V', E')$ as follows.  First we attach two vertices $s$ and $t$ to $G$.
Select one vertex, say $s$, from the clique and connect each vertex in $G$ to $s$.
We connect the non-connected vertices with edges of weight $0$.
Let $P = \abs{V'} - 1$.
We will weight the edges in $G$ with $1$, let us define $\alpha = P - 0.5/P^2$.
Set the weight of an edge $\wght{(s, n)} = \alpha - \deg\pr{n}$, for each $n \in V$.
Due to this scheme we have $\sum_{(n, y) \in E'} \wght{(n, y)} = \alpha$ for any $n \in V$.
Finally, we set $\wght{(s, t)} = \abs{V'}\alpha$. This weight is so large that no solution
for \textsc{SparseNbhd} will contain $s$ or $t$.

Let $X$ be a clique of size $N$ in $G$. Then the weight of $X$ is equal to
\[
	\dnst{X, V'} = \frac{\alpha N - {N \choose 2}}{PN - {N \choose 2}} = \frac{2\alpha - N + 1}{2P - N + 1}\quad.
\]

If we have a non-clique subgraph of size $N$, then obviously its weight is
genuinely larger than $\dnst{X, V'}$.

Assume a set $Z \subseteq V$ with $K < N$ vertices. The weight of $Z$ is bounded by
\[
    \dnst{Z, V'} \geq  \frac{\alpha K - {K \choose 2}}{PK - {K \choose 2}}= \frac{2\alpha - K + 1}{2P - K + 1}\quad.
\]
We want $\dnst{X, V'} < \dnst{Z, V'}$, which is guaranteed if
\begin{equation}
\label{eq:bs1}
	\frac{2\alpha - N + 1}{2P - N + 1} <  \frac{2\alpha - K + 1}{2P - K + 1}\quad.
\end{equation}
If we have a non-clique subgraph of size $N$, then obviously its weight is
genuinely smaller than $\dnst{X, X \cup S}$.

Since $-K + 1 > -N + 1$, Lemma~\ref{lem:ratio} implies that if
\[
	2\alpha < 2P + \frac{2P - N + 1}{(N - 1) - (K - 1)}0 = 2P,
\]
then the inequality in Eq~\ref{eq:bs1} is guaranteed. This is guaranteed by our choice of $\alpha$.

Let $Y \subseteq V$ be a non-clique of size $M > N$ in $G$. Then the weight of $Y$ bounded by
\[
	\dnst{Y, V'} \geq  \frac{\alpha M - {M \choose 2} + 1}{PM - {M \choose 2}}= \frac{2\alpha + 2/M - M + 1}{2P - M + 1}\quad.
\]
We need to have $\dnst{X, V'} < \dnst{Y, V'}$, which is guaranteed if
\begin{equation}
\label{eq:bs2}
	\frac{2\alpha - N + 1}{2P - N + 1} < \frac{2\alpha + 2/M - M + 1}{2P - M + 1}\quad.
\end{equation}

Since $-M + 1 < -N + 1$, Lemma~\ref{lem:ratio} guarantees that if
\[
	2\alpha >  2P + \frac{2}{M}\frac{2P - N + 1}{(N - 1) - (M - 1)} = 2P - \frac{2(2P - N + 1)}{M(M - N)}
\]
then Eq.~\ref{eq:bs2} is guaranteed. This is guaranteed by our choice of $\alpha$.

Let $k$ be the minimum size of the clique given as a parameter in \textsc{Clique}. Set $\beta = \frac{2\alpha - k + 1}{2P - k + 1}$.
If $G$ contains a clique of size $k$, then there is a subgraph in $G'$ with a density of $\beta$.
Assume now that $G'$ contains a subgraph, say $H$, with a density of at most $\beta$.
Note that $\beta$ is largest, when $k = 1$, that is, $\beta \leq \alpha / P$.
If $s$ or $t$ is contained in $H$, then the density is at least $2\wght{(s, t)} / P(P + 1) > \alpha / P$,
which is a contradiction. Hence $H$ is a subgraph of $G$.  $H$ must contain
at least $k$ vertices, otherwise bound in Eq.~\ref{eq:bs1} is violated. $H$ must be a clique, otherwise
bound in Eq.~\ref{eq:bs2} is violated. Consequently, $G$ has a clique of size $k$ if and only if
$G'$ has a subgraph of density at least $\beta$. The reduction is polynomial.
This concludes the proof.\qed
\end{proof}
\fi

\section{Related Work}
\label{sec:related}

Finding communities in graphs and social networks is one of the most
well-studied topics in graph mining.
The amount of literature on the subject is very extensive. 
This section cannot aspire to cover all the different approaches and
aspects of the problem, we only provide a brief overview of the area.

\textbf{Community detection.}
A large part of the related work deals with the problem of
partitioning a graph in disjoint clusters or communities. 
A number of different methodologies have been applied, such as 
hierarchical approaches~\cite{girvan2002community},
methods based on modularity maximization~\cite{agarwal2008modularity, clauset2004finding,
  girvan2002community, white2005spectral},
graph-theoretic approaches~\cite{flake2000efficient,flake2002self},
random-walk methods~\cite{DBLP:journals/jgaa/PonsL06,mcl,conf/iccS/ZhouL04}, 
label-propagation approaches~\cite{mcl}, 
and 
spectral graph
partition~\cite{chung1997spectral,karypis98multilevel,ng2001spectral,DBLP:journals/sac/Luxburg07}.
A thorough review on community-detection methods can be found on the
survey by Fortunato~\cite{fortunato2010community}.
We note that this line of work is different than the present paper,
since we do not aim at partitioning a graph in disjoint communities. 

\textbf{Overlapping communities.}
Researchers in community detection have realized that, 
in many real situations and real applications, 
it is meaningful to consider that graph vertices do not belong only to
one community. 
Thus, one asks to partition a graph into overlapping communities. 
Typical methods here rely on 
clique percolation~\cite{palla2005uncovering},
extensions to the modularity-based
approaches~\cite{DBLP:conf/pkdd/Gregory07, pinney2006betweenness}, 
analysis of ego-networks~\cite{DBLP:conf/kdd/CosciaRGP12},
or fuzzy clustering~\cite{Zhang2007overlapping}.
Again the problem we address in this paper is quite different. 
First, we find communities centered around a given set of source
vertices, and not for the whole graph.
Second, the communities output by our algorithm do not have arbitrary
overlaps, but they have a specific nested structure.

\textbf{Centerpiece subgraphs and community search.}
Perhaps closer to our approach is work related to the centerpiece
subgraphs and the community-search
problem~\cite{DBLP:conf/kdd/TongF06,DBLP:journals/tkdd/KorenNV07,DBLP:conf/kdd/SozioG10}.
In this class of problems, a set of source vertices $S$ is given and
the goal is to find a subgraph so that $S$ belongs in the
subgraph and the subgraph forms a tight community. 
The quality of the subgraph is measured with various objective
functions, such as degree~\cite{DBLP:conf/kdd/SozioG10}, 
conductance~\cite{DBLP:journals/tkdd/KorenNV07}, or 
random-walk-based measures~\cite{DBLP:conf/kdd/TongF06}.
The difference of these methods with the one presented here is that
these methods return only one community, while in this paper we
deal with the problem of finding a sequence of nested communities. 
 
\smallskip
In summary, despite the numerous research on the topic of community
detection in graphs and social networks, to the best of our knowledge,
this is the first paper to address the topic of nested communities
with respect to a set of source vertices. 
Furthermore, our approach offers novel technical ideas, such as
providing a solid theoretical analysis that allows to decompose the
problem of finding nested communities into two sub-problems:
($i$) ordering the set of vertices, and ($ii$) 
segmenting the graph vertices according to that given order.

\section{Experimental Evaluation}\label{sec:exps}

We will now provide experimental evidence that  our method efficiently discovers meaningful segmentations and that our ordering algorithm 
outperforms several natural baselines.

\textbf{Datasets and experimental setup.}
In our experiments we used six datasets, 
five obtained from Mark Newman's webpage,\!\footnote{\url{http://www-personal.umich.edu/~mejn/netdata/}}
and a bibliographic dataset obtained from DBLP.
The datasets are as follows:
{\em Adjnoun}:
adjacency graph of common adjectives and nouns in the novel David
Copperfield, 
by Charles Dickens.
{\em Dolphins}:
an undirected social graph of frequent associations between 62
dolphins in a community living off Doubtful Sound, New Zealand. 
{\em Karate}:
social graph of friendships between 34 members of a karate club at a
US university in the 1970s.
{\em Lesmis}:
coappearance graph of characters in the novel Les Miserables. 
{\em Polblogs}:
a directed graph of hyperlinks between weblogs on US politics,
recorded in 2005.
{\em DBLP}:
coauthorship graph between researchers in computer science. 
The statistics of these datasets are given in Table~\ref{tab:stats}.

For each dataset and a given source set $S$, we considered three different weighting schemes:
First we run personalized PageRank using the source node with a restart of $0.1$. Let $p(v)$
be the PageRank weight of each vertex. Given an edge $e = (v, w)$, we set three different weighting schemes,
\[
	\wghtnorm{e} = \frac{p(v)}{\degree{v}} + \frac{p(w)}{\degree{w}},\quad
	\wghtsum{e} = p(v) + p(w),\quad 
	\wghtmin{e} = \min(p(v), p(w)).
\]
These weights are selected so that the vertices that are hard to reach with a random walk
will have edges with small weights, and hence will be placed in outer communities. 
For \emph{DBLP}, we weighted the edges during PageRank computation with the number of joint papers, each paper normalized by the number of authors.
We use the vertex with the highest degree as a starting set.

\begin{table}[t]
\caption{Basic statistics of graphs (first two columns) and performance over hops baseline. The third column represents a typical running time while the fourth
column represents a typical number of entries during the segmentation. The last three columns represent the normalized score compared to the baseline score $\score{\fm{H}}$. 
}
\label{tab:stats}
\begin{tabular*}{\linewidth}{@{\extracolsep{\fill}}lrr@{\hspace{17pt}}rr@{\hspace{17pt}}rrr}
\toprule
&&&&&\multicolumn{3}{l}{performance $\score{\fm{V}} / \score{\fm{H}}$}
\\
\cmidrule{6-8}
Name & $\abs{V(G)}$ & $\abs{E(G)}$ & Time & $N$ &
$\wghtnorm{}$ & $\wghtsum{}$ & $\wghtmin{}$ 
\\
\midrule
\emph{Adjnoun} & 
112 & 425 &
2ms & 84 &
${0.90} / {0.95}$ & ${0.88} / {0.95}$ & ${0.77} / {0.94}$

\\
\emph{Dolphins} &
62 & 159 &
1ms & 41 &
${0.67} / {0.80}$ & ${0.61} / {0.78}$ & ${0.57} / {0.80}$

\\
\emph{Karate} &
34 & 78 &
1ms & 21 &
${0.78} / {0.91}$ & ${0.76} / {0.91}$ & ${0.60} / {0.93}$

\\
\emph{Lesmis} &
77 & 254 &
2ms & 37 &
${0.77} / {0.93}$ & ${0.84} / {0.94}$ & ${0.62} / {0.94}$

\\
\emph{Polblogs} &
$1\,222$ & $16\,714$ &
84ms & 872 &
${0.87} / {0.96}$ & ${0.95} / {0.99}$ & ${0.57} / {0.96}$

\\
\emph{DBLP} &
$703\,193$ & $2\,341\,362$ &
23s & $1\,797$ &
${0.87} / {0.99}$ & ${0.98} / {1.00}$ & ${0.45} / {0.99}$
\\
\bottomrule
\end{tabular*}
\end{table}

\textbf{Time complexity.} 
Our first step is to study the running time of our algorithm.  
We ran our experiments on a laptop equipped with a
1.8 GHz dual-core Intel Core i7 with 4 MB shared L3 cache,
and typical running times for
each dataset are given in 3rd column of
Table~\ref{tab:stats}.\footnote{For the code, see   \url{http://users.ics.aalto.fi/~ntatti/}}   
Our algorithm is fast: 
for the largest dataset with 2 million edges, the computation took only 20 seconds.
The algorithm consists of 4 steps, computing PageRank, ordering the vertices,
grouping the vertices into blocks such that monotonicity condition is
guaranteed, and segmenting the groups. The only computationally strenuous step
is segmentation which requires quadratic time in the number of blocks.
The number of vertices in \emph{DBLP} is over $700\,000$,
however, grouping according to the PAV algorithm leaves only $2\,000$
blocks, 
which can be easily segmented. 
It is possible to select weights in such a way that there
will no reduction when grouping vertices, so that finding the optimal
segmentation becomes infeasible. However, in such a case, we can always resort
to a near-linear approximation optimization algorithm~\cite{DBLP:journals/tods/GuhaKS06}.

\textbf{Comparison to baseline.} A key part in our approach is discovering a
good order.  Our next step is to compare the order induced by \sortnodes
against several natural baselines. For the first baseline we group the vertices based
on the length of a minimal path from the source. We then compared 
these communities, say $\fm{H}$, to the (same number of) communities obtained with our method.  The scores, given in Table~\ref{tab:stats}, show that
our approach beats this baseline in every case, which is expected since this
na\"{i}ve baseline does not take into account density. For our next two baselines we order
vertices based on vertex degree and PageRank. We then compute community sequences
with $2$--$10$ communities from these orders.  Typical scores are given in
Figure~\ref{fig:baseline}.  Out of $6 \times 3 \times 9 = 162$ comparisons,
\sortnodes wins both orders 158 times, ties once 
(\emph{Karate}, \wghtmin{}, 3 communities) and loses 3 times to the degree order
(\emph{DBLP}, \wghtnorm{}, 3--5 communities).

\begin{figure}[t]
\begin{center}
\begin{tikzpicture}
\begin{axis}[xlabel={number of communities}, ylabel= {$\score{\fm{V}} / \score{\fm{B}}$}, title = {weight \wghtnorm{}},
    width = 4.3cm,
    cycle list name=yaf,
	legend to name=leg:score,
	legend columns = 3,
	legend entries = {\sortnodes, \textsc{Degree}, \textsc{PageRank}}
    ]
\addplot table[x index = 0, y index = 37, header = false] {scores.dat};
\addplot table[x index = 0, y index = 38, header = false] {scores.dat};
\addplot table[x index = 0, y index = 39, header = false] {scores.dat};
\pgfplotsextra{\yafdrawaxis{2}{10}{0.86}{0.93}}
\end{axis}
\end{tikzpicture}%
\begin{tikzpicture}
\begin{axis}[xlabel={number of communities},title = {weight \wghtsum{}},
    width = 4.3cm,
    cycle list name=yaf,
    ]
\addplot table[x index = 0, y index = 40, header = false] {scores.dat};
\addplot table[x index = 0, y index = 41, header = false] {scores.dat};
\addplot table[x index = 0, y index = 42, header = false] {scores.dat};
\pgfplotsextra{\yafdrawaxis{2}{10}{0.94}{0.973}}
\end{axis}
\end{tikzpicture}%
\begin{tikzpicture}
\begin{axis}[xlabel={number of communities},title = {weight \wghtmin{}},
    width = 4.3cm,
    cycle list name=yaf,
    ]
\addplot table[x index = 0, y index = 43, header = false] {scores.dat};
\addplot table[x index = 0, y index = 44, header = false] {scores.dat};
\addplot table[x index = 0, y index = 45, header = false] {scores.dat};
\pgfplotsextra{\yafdrawaxis{2}{10}{0.54}{0.75}}
\end{axis}
\end{tikzpicture}

\ref{leg:score}
\end{center}

\caption{Quality scores of community sequences based on different orders as a function of number of communities for \emph{Polblogs}. The scores are normalized
by the score of a community sequence $\fm{B}$ with a single community.}
\label{fig:baseline}
\end{figure}
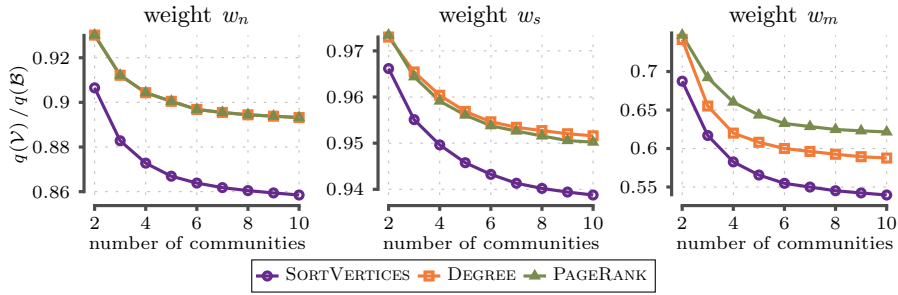

\textbf{Examples of Communities.}
Our final step is to provide examples of discovered communities.  In
Figure~\ref{fig:karate} we provide 4 different community sequences with 3 communities
using weights \wghtsum{} and \wghtnorm{} and sources $S = \set{1}$ and $S = \set{33, 34}$.
The inner-most community for $1$ contains a near 5-clique.
The inner-most community for $33, 34$ contains two 4-cliques.
The normalized weight \wghtnorm{} penalizes hubs. This can be seen in Figure~\ref{fig:karate:a},
where hubs $33$, $34$ move from the outer community to the middle community.
Similarly, hub $1$ changes communities in Figure~\ref{fig:karate:b}.
Finally, we give an example of communities discovered in \emph{DBLP}.
Table~\ref{tab:dblpex} contains communities discovered around Christos Papadimitriou.
Authors in inner communities share many joint papers with Papadimitriou.

\begin{figure}[t]
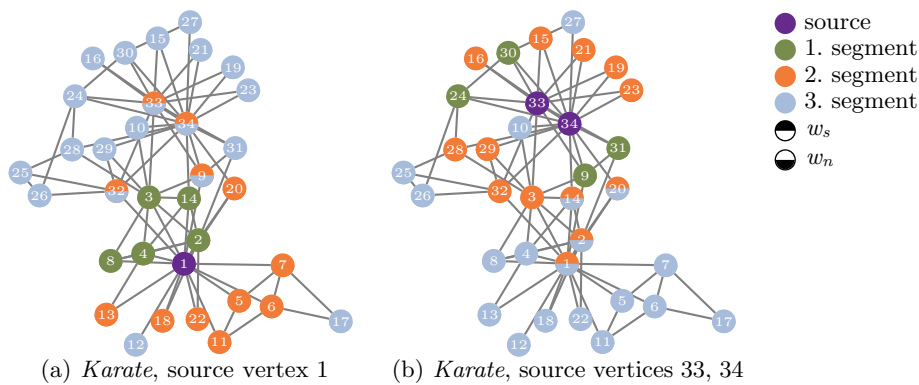

\subfigure[\emph{Karate}, source vertex 1\label{fig:karate:a}]{\input{karate_1}}\hfill
\subfigure[\emph{Karate}, source vertices 33, 34\label{fig:karate:b}]{\input{karate_3334}}\hfill
\begin{tikzpicture}[baseline = (current bounding box.north)]
\node[circle, line width = 0pt, text = white, inner sep = 0.5pt, text width = 7pt, fill = yafcolor1, label = {right:source}] (n1) {};
\node[circle, line width = 0pt, text = white, inner sep = 0.5pt, text width = 7pt, fill = yafcolor3, label = {right:1. segment}, below = 2pt of n1] (n2) {};
\node[circle, line width = 0pt, text = white, inner sep = 0.5pt, text width = 7pt, fill = yafcolor2, label = {right:2. segment}, below = 2pt of n2] (n3) {};
\node[circle, line width = 0pt, text = white, inner sep = 0.5pt, text width = 7pt, fill = yafcolor5!40, label = {right:3. segment}, below = 2pt of n3] (n4) {};
\node[circle, fill, line width = 1pt, text = white, inner sep = 0.5pt, text width = 7pt, circle split part fill = {black, white}, label = {right:\wghtsum{}}, below = 2pt of n4] (n5) {};
\node[circle, fill, line width = 1pt, text = white, inner sep = 0.5pt, text width = 7pt, circle split part fill = {white, black}, label = {right:\wghtnorm{}}, below = 2pt of n5] (n6) {};
\end{tikzpicture}
\caption{4 community sequences with 3 communities of \emph{Karate}.  Segmentations in
Figure~\ref{fig:karate:a} use 1 as a source and community sequences in
Figure~\ref{fig:karate:b} use 33, 34 as sources. Communities are
decoded as colors, the top-half represents \wghtsum{}, the bottom-half represents \wghtnorm{}.}
\label{fig:karate}
\end{figure}

\begin{table}[t]
\label{tab:dblpex}
\caption{Top-3 communities from a sequence of 5 communities for Christos Papadimitriou from \emph{DBLP} set and using \wghtsum{}.}
\scriptsize
\begin{tabular*}{\linewidth}{@{\extracolsep{\fill}}*{6}{l}}
\toprule
\textbf{1. segment} & D. Johnson & E. Dahlhaus & V. Vianu & G. Gottlob & A. Itai \\
M. Yannakakis & M. Garey & P. Crescenzi & P. Kanellakis & M. Sideri & A. Sch\"{a}ffer \\
F. Afrati & R. Karp & P. Seymour & S. Abiteboul & E. Koutsoupias & A. Aho \\
\textbf{2. segment} & R. Fagin & O. Vornberger & A. Piccolboni & C. Daskalakis & P. Serafini \\
J. Ullman & \textbf{3. segment} & M. Blum & D. Goldman & X. Deng & P. Raghavan \\
Y. Sagiv & G. Papageorgiou & K. Ross & E. Arkin & P. Goldberg & P. Bernstein \\
S. Cosmadakis & V. Vazirani & P. Kolaitis & I. Diakonikolas & T. Hadzilacos \\
\bottomrule
\end{tabular*}
\end{table}

\section{Concluding Remarks}\label{sec:concl}
We considered a problem of discovering nested communities, a sequence of subgraphs
such that each community is a more connected subgraph of the next community.
We approach the problem by dividing it into two subproblems: discovering
the community sequence for a fixed order of vertices, a problem which we can solve
efficiently, and discovering an order. We provided a simple heuristic for
discovering an order, and provided theoretical and empirical evidence that
this order is good.

Discovering nested communities seems to have a lot of potential as it is
possible to modify or extend the problem in many ways. We can generalize the
problem by not only considering sequences but, for example, trees of communities,
where a parent node needs to be a denser subgraph than the child
node. 
Another possible extension is to consider multiple source sets instead
of just one.

\subsubsection*{Acknowledgements.} This work was supported by Academy of Finland grant 118653 ({\sc algodan})

\bibliographystyle{plain-initials}
\bibliography{bibliography,communities}

\begin{thebibliography}{10}

\bibitem{agarwal2008modularity}
G.~Agarwal and D.~Kempe.
\newblock Modularity-maximizing network communities via mathematical
  programming.
\newblock {\em European Physics Journal B}, 66(3), 2008.

\bibitem{PAV}
M.~Ayer, H.~Brunk, G.~Ewing, and W.~Reid.
\newblock An empirical distribution function for sampling with incomplete
  information.
\newblock {\em The annals of mathematical statistics}, 26(4), 1955.

\bibitem{Bellman}
R.~Bellman.
\newblock On the approximation of curves by line segments using dynamic
  programming.
\newblock {\em Communications of the ACM}, 4(6), 1961.

\bibitem{Charikar}
M.~Charikar.
\newblock Greedy approximation algorithms for finding dense components in a
  graph.
\newblock In {\em APPROX}, 2000.

\bibitem{chung1997spectral}
F.~R.~K. Chung.
\newblock {\em Spectral Graph Theory}.
\newblock American Mathematical Society, 1997.

\bibitem{clauset2004finding}
A.~Clauset, M.~E.~J. Newman, , and C.~Moore.
\newblock Finding community structure in very large networks.
\newblock {\em Physical Review E}, 2004.

\bibitem{DBLP:conf/kdd/CosciaRGP12}
M.~Coscia, G.~Rossetti, F.~Giannotti, and D.~Pedreschi.
\newblock {DEMON}: a local-first discovery method for overlapping communities.
\newblock In {\em KDD}, 2012.

\bibitem{flake2000efficient}
G.~W. Flake, S.~Lawrence, and C.~L. Giles.
\newblock Efficient identification of web communities.
\newblock In {\em KDD}, 2000.

\bibitem{flake2002self}
G.~W. Flake, S.~Lawrence, C.~L. Giles, and F.~M. Coetzee.
\newblock Self-organization and identification of web communities.
\newblock {\em Computer}, 35(3), 2002.

\bibitem{fortunato2010community}
S.~Fortunato.
\newblock Community detection in graphs.
\newblock {\em Physics Reports}, 486, 2010.

\bibitem{girvan2002community}
M.~Girvan and M.~E.~J. Newman.
\newblock Community structure in social and biological networks.
\newblock {\em PNAS}, 99, 2002.

\bibitem{DBLP:conf/pkdd/Gregory07}
S.~Gregory.
\newblock An algorithm to find overlapping community structure in networks.
\newblock In {\em PKDD}, 2007.

\bibitem{DBLP:journals/tods/GuhaKS06}
S.~Guha, N.~Koudas, and K.~Shim.
\newblock Approximation and streaming algorithms for histogram construction
  problems.
\newblock {\em ACM TODS}, 31, 2006.

\bibitem{unimodal}
N.~Haiminen and A.~Gionis.
\newblock Unimodal segmentation of sequences.
\newblock In {\em ICDM}, 2004.

\bibitem{karypis98multilevel}
G.~Karypis and V.~Kumar.
\newblock Multilevel algorithms for multi-constraint graph partitioning.
\newblock In {\em CDROM}, 1998.

\bibitem{DBLP:journals/tkdd/KorenNV07}
Y.~Koren, S.~C. North, and C.~Volinsky.
\newblock Measuring and extracting proximity graphs in networks.
\newblock {\em TKDD}, 1(3), 2007.

\bibitem{DBLP:conf/www/LeskovecLDM08}
J.~Leskovec, K.~J. Lang, A.~Dasgupta, and M.~W. Mahoney.
\newblock Statistical properties of community structure in large social and
  information networks.
\newblock In {\em WWW}, 2008.

\bibitem{ng2001spectral}
A.~Y. Ng, M.~I. Jordan, and Y.~Weiss.
\newblock On spectral clustering: Analysis and an algorithm.
\newblock In {\em NIPS}, 2001.

\bibitem{palla2005uncovering}
G.~Palla, I.~Der\'enyi, I.~Farkas, and T.~Vicsek.
\newblock Uncovering the overlapping community structure of complex networks in
  nature and society.
\newblock {\em Nature}, 435, 2005.

\bibitem{pinney2006betweenness}
J.~Pinney and D.~Westhead.
\newblock Betweenness-based decomposition methods for social and biological
  networks.
\newblock In {\em Interdisciplinary Statistics and Bioinformatics}, 2006.

\bibitem{DBLP:journals/jgaa/PonsL06}
P.~Pons and M.~Latapy.
\newblock Computing communities in large networks using random walks.
\newblock {\em Journal of Graph Algorithms Applications}, 10(2), 2006.

\bibitem{DBLP:conf/kdd/SozioG10}
M.~Sozio and A.~Gionis.
\newblock The community-search problem and how to plan a successful cocktail
  party.
\newblock In {\em KDD}, 2010.

\bibitem{DBLP:conf/kdd/TongF06}
H.~Tong and C.~Faloutsos.
\newblock Center-piece subgraphs: problem definition and fast solutions.
\newblock In {\em KDD}, 2006.

\bibitem{mcl}
S.~van Dongen.
\newblock {\em Graph Clustering by Flow Simulation}.
\newblock PhD thesis, University of Utrecht, 2000.

\bibitem{DBLP:journals/sac/Luxburg07}
U.~von Luxburg.
\newblock A tutorial on spectral clustering.
\newblock {\em Statistics and Computing}, 17(4), 2007.

\bibitem{white2005spectral}
S.~White and P.~Smyth.
\newblock A spectral clustering approach to finding communities in graph.
\newblock In {\em SDM}, 2005.

\bibitem{Zhang2007overlapping}
S.~Zhang, R.-S. Wang, and X.-S. Zhang.
\newblock Identification of overlapping community structure in complex networks
  using fuzzy $c$-means clustering.
\newblock {\em Physica A}, 2007.

\bibitem{conf/iccS/ZhouL04}
H.~Zhou and R.~Lipowsky.
\newblock Network brownian motion: A new method to measure vertex-vertex
  proximity and to identify communities and subcommunities.
\newblock In {\em ICCS}, 2004.

\end{thebibliography}

\end{document}